\documentclass[11pt]{article}

\usepackage{amssymb}

\usepackage[latin1]{inputenc}
\usepackage[english]{babel}
\usepackage{amsmath}
\usepackage{amsfonts}
\usepackage{graphicx}
\usepackage{amssymb}
\usepackage{fancyvrb}
\usepackage{wrapfig}
\usepackage{color}
\usepackage{enumerate}
\usepackage{hyperref}
\usepackage[margin=2cm]{geometry}
\usepackage{amsthm}
\usepackage{environ}

\theoremstyle{plain}
\newtheorem{theorem}{Theorem}

\newtheorem{lemma}[theorem]{Lemma}
\newtheorem{claim}{Claim}

\theoremstyle{definition}
\newtheorem{definition}[theorem]{Definition}

\newcommand{\diam}{\, \mathrm{diam}}
\newcommand{\perim}{\, \mathrm{perim}}
\newcommand{\Mitchell}{{Mitchell}}
\newcommand{\vare}{\varepsilon}
\newcommand{\Wlog}{without loss of generality}
\newcommand{\dd}{\,\mathrm{d}}
\newcommand{\set}[1]{\{#1\}}
\newcommand{\conv}{\textnormal{conv}}
\newcommand{\E}{$_E$ }

\usepackage[explicit]{titlesec}

\titleformat{\section}
  {\Large}{\textbf{\thesection}}{1em}{\textbf{#1}}

\titleformat{\subsection}
  {}{\textbf{\thesubsection}}{1em}{\textbf{#1}}

\begin{document}

\title{The guillotine subdivision approach for TSP with neighborhoods revisited}

\author{Sophie Spirkl\\ \normalsize Research Institute for Discrete Mathematics, Lenn\'estr. 2,
  53113 Bonn, Germany \\
\href{mailto:spirkl@or.uni-bonn.de}{\texttt{\emph{spirkl@or.uni-bonn.de}}}}

\maketitle

\hrulefill

\begin{abstract}
  The Euclidean TSP with neighborhoods (TSPN) is the following
  problem: Given a set $\mathcal{R}$ of $k$ regions (subsets of
  $\mathbb{R}^2$), find a shortest tour that visits at least one point
  from each region. We study the special cases of disjoint, connected,
  $\alpha$-fat regions (i.e., every region $P$ contains a disk of
  diameter $\frac{\diam(P)}{\alpha}$) and disjoint unit disks. 

  For the latter, Dumitrescu and Mitchell \cite{dumi} proposed an
  algorithm based on Mitchell's guillotine subdivision approach for
  the Euclidean TSP \cite{mitchell-tsp}, and claimed it to be a PTAS.
  However, their proof contains a severe gap, which we will close in
  the following. Bodlaender et al. \cite{grigoriev} remark that their
  techniques for the minimum corridor connection problem based on
  Arora's PTAS for TSP \cite{arora} carry over to the TSPN and yield
  an alternative PTAS for this problem.

  For disjoint connected $\alpha$-fat regions of varying size,
  Mitchell \cite{mitchell-ptas} proposed a slightly different PTAS
  candidate. We will expose several further problems and gaps in this
  approach. Some of them we can close, but overall, for $\alpha$-fat
  regions, the existence of a PTAS for the TSPN remains open.

\vspace*{0.3cm}

\end{abstract}

\hrulefill

\section{TSP among $\alpha$-fat Regions}

\subsection{Problem Definition and Background}

The Euclidean TSP with neighborhoods (TSPN) is the following problem:
Given a set $\mathcal{R}$ of $k$ regions (subsets of $\mathbb{R}^2$), find a shortest
tour that visits at least one point from each region.

Even for disjoint or connected regions, the TSPN does not admit a PTAS
unless $P=NP$ \cite{safra}. Aiming for a PTAS under additional
restrictions on the input, \cite{mitchell-ptas} and \cite{elb} require
connected and disjoint regions, and both introduce a notion of
$\alpha$-fatness. 

\begin{definition}[\cite{mitchell-ptas}] \label{def:afat-mit} A region
  $P$ of points in the plane is \textbf{$\boldsymbol{\alpha}$-fat}, if
  it contains a disk of diameter $\dfrac{\mathrm{diam}(P)}{\alpha}$.
\end{definition}

\begin{definition}[\cite{elb}] \label{def:afat-elb} A region $P$ in
  the plane is \textbf{$\boldsymbol{\alpha}$-fat$_{\boldsymbol{E}}$},
  if for every disk $\Theta$, such that the center of $\Theta$ is
  contained in $P$ but $\Theta$ does not fully contain $P$, the area
  of the intersection $P \cap \Theta$ is at least $\frac{1}{\alpha}$
  times the area of $\Theta$.
\end{definition}

For $\alpha$-fat\E regions, Chan and Elbassioni \cite{elb-qptas}
developed a quasi-polynomial time approximation scheme (even for a
more general notion of fatness and in more general metric
spaces). Mitchell \cite{mitchell-ptas} was the first to consider
$\alpha$-fat regions. Bodlaender et al. \cite{grigoriev} introduced
the notion of geographic clustering, where each region contains a
square of size $q$ and has diameter at most $cq$ for a fixed constant
$c$, which is a special case of $\alpha$-fat regions. They showed that
the TSPN with geographic clustering admits a PTAS based on Arora's
framework for the Euclidean TSP.

In all cases, $\alpha$-fatness provides a lower bound (in terms of
their diameters) on the length of a tour visiting disjoint regions, but in
the following, the second definition will turn out to be more useful. 
Throughout this paper, $\alpha \geq 1$ and $\vare > 0$ will be
constants. 

\subsection{Mitchell's Algorithm}

The core of Mitchell's algorithm is dynamic programming, which
requires certain restrictions on the space of solutions. To this end,
Mitchell claims the following: 

There is an almost optimal tour (up to a factor of $1 + \vare$) such that:
\begin{enumerate}[(A)]
\item The tour visits the minimum-diameter axis-aligned rectangle
  $R_0$ intersecting all regions, and therefore has to be located
  within a window $W_0$ of diameter $\mathcal{O}(\diam(R_0))$
  intersecting $R_0$. We distinguish internal regions
  $\mathcal{R}_{W_0}$ that are entirely contained in $W_0$, and
  external regions.
\item We can require the vertices of the tour
  to lie on a polynomial-size grid (in $k$ and $\frac{1}{\vare}$)
  within this rectangle.
\item The tour is a connected Eulerian graph
  fulfilling the ``$(m,M)$-guillotine property" (which roughly states
  that there is a recursive decomposition of the bounding box of the
  tour by cutting it into subwindows such that the structure of
  internal regions and tour segments on the cut is of bounded
  complexity in $m$ and $M$), again at a loss of only $\vare$ for
  appropriately chosen $m$ and $M$.
\item The tour obeys (B) and (C) simultaneously. 
\item The external regions can be dealt with efficiently as there is
  only a polynomial number of ways for them to be visited by an
  $(m,M)$-guillotine tour (i. e. for every cut, there is a polynomial
  number of options for which regions will be visited on which side of
  it). 
\end{enumerate}

Under these assumptions, Mitchell states a dynamic programming
algorithm. Starting with a window (axis-parallel rectangle) $W_0$,
which is assumed to contain all edges of the tour, every subwindow $W$
defines several subproblems (see Figure~\ref{pic:sub}). The
subproblems also enumerate all possible configurations of edge
segments intersecting its boundary, connection patterns of these
segments, internal (contained in $W_0$) and external (intersecting
$\partial W_0$) regions to be visited inside and outside of the
window, cuts (horizontal or vertical lines dividing $W$ into two
subwindows) and configurations on the cut. For each cut, the
subproblems to both sides will already have been solved through a
bottom-up recursion, therefore we can select an optimal solution with
compatible configurations. The optimum (shortest) solution for the
subproblem (among all possible cuts) is stored and can be used for the
next recursion level.

\begin{figure}[hbt]
  \centering
  \includegraphics[width=7cm]{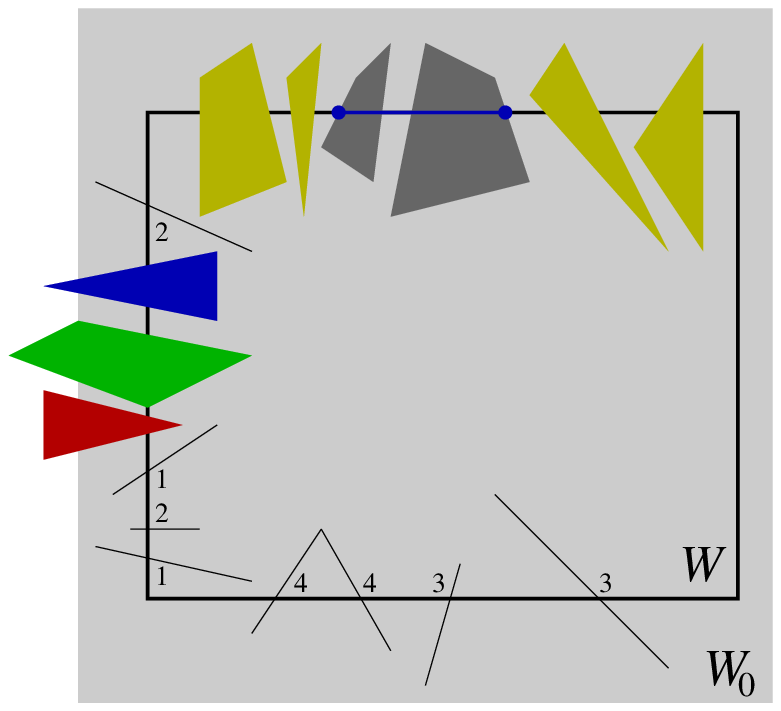}
  \caption{Structure of a subproblem}
  \label{pic:sub}
\end{figure}

Assumption A is false, and will be rectified in Section~\ref{sec:loc},
Lemma~\ref{lem:loc}. Statement B is correct. For the third statement,
a stronger assumption on the regions can be used to mend the upper
bound for the additional length incurred in Mitchell's construction;
see Section~\ref{sec:guillotine}, Theorem~\ref{thm:charge}. Preserving
connectivity in a graph with guillotine property is difficult, not
accounted for in \cite{mitchell-ptas} and for Mitchell's line of
argument not clear. We present a counterexample in
Section~\ref{sec:cnn}, Figure~\ref{pic:region-span}. While not proven
by Mitchell, statement D is still correct (if assumption C holds for
the given tour), a technical
argument will be sketched in Section~\ref{sec:grid}. The last
statement is again false, but can be fixed using a different notion of
$\alpha$-fatness, which we will show in Section~\ref{sec:ext}.

\subsection{Localization} \label{sec:loc}

In \Mitchell's algorithm, the search for a (nearly) optimal tour among
a set $\mathcal{R}$ of connected regions is restricted to a small
neighborhood of the minimum-diameter axis-aligned rectangle $R_0$ that
intersects all regions.

\begin{claim}[{\cite[Lemma 2.4]{mitchell-ptas}}] 
  There exists an optimal tour $T^*$ of the regions in $\mathcal{R}$
  that lies within the ball $B(c_0, 2 \diam (R_0))$ of radius $2 \diam
  (R_0)$ around the center point $c_0$ of $R_0$.
\end{claim}

However, Figure~\ref{pic:loc} shows that in general, this is false: a
nearly optimal tour need not be within $\mathcal{O}(\diam R_0)$
distance of $R_0$, even if the regions are $\alpha$-fat, disjoint and
connected as in Figure~\ref{pic:loc}. The vicinity of $R_0$ only
contains a $\sqrt{2}$-approximation of the optimum, which is instead
found within $R_1$.

\begin{figure}[hbt] \center
  \includegraphics[width=3.8cm]{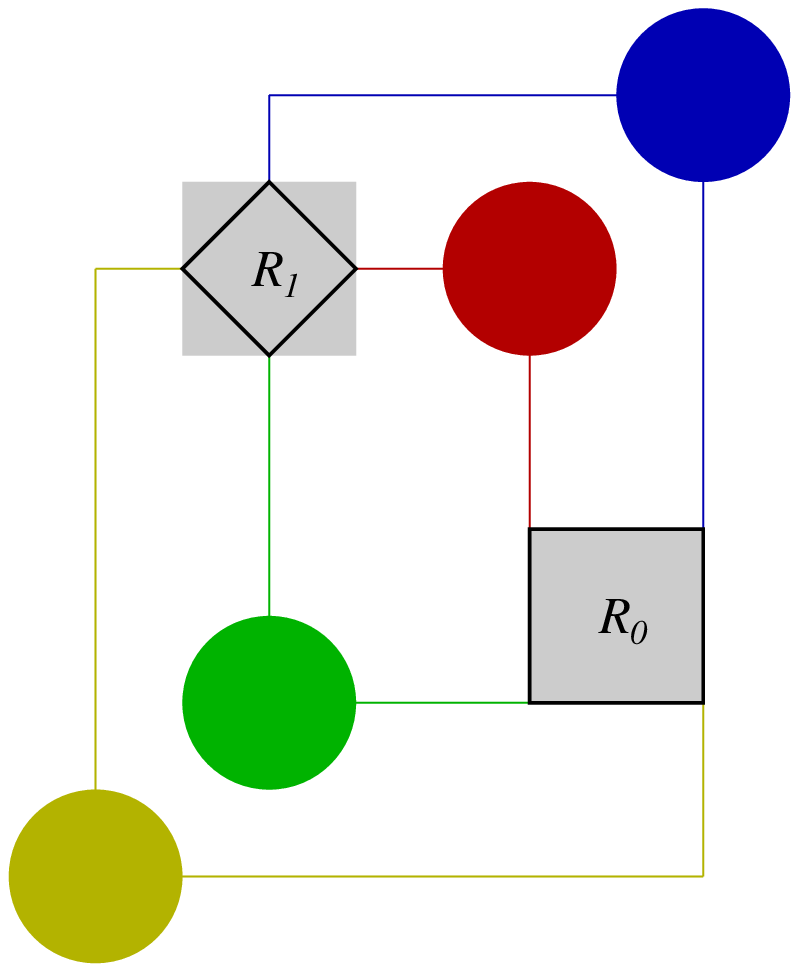} 
  \caption{Localization of an optimal tour}
  \label{pic:loc}
\end{figure}

We now show how this problem can be resolved: If an optimal tour
intersected $R_0$, \Mitchell's lemma would be correct. He argues that,
if some regions were to be visited far away from $R_0$, the path
leading to them could be replaced by $\partial R_0$, which due to
connectivity must visit those regions. Otherwise, no region can be
fully contained in $R_0$, so the same argument yields that every
region must intersect $\partial R_0$, making $\perim(R_0) \leq
2\sqrt{2}\diam(R_0)$ an upper bound for the length $L^*$ of an optimal
solution. Combining this with the fact that $L^* \geq 2\diam(R_0)$,
$L^*$ is now known up to a constant factor.

Now, there are two cases to consider: If there is a small region (of
diameter $\mathcal{O}(L^*)$), an area of diameter $\mathcal{O}(L^*)$
around this region must contain an optimal tour. Otherwise, all
regions are of diameter $> \mathcal{O}(L^*)$. If the regions are
required to be polygons, it is possible to limit the number of
possible (approximate) locations of an optimal tour by adapting an
approach by J. Gudmundsson and C. Levcopoulos \cite[Section
5.1]{gudm}, who show that in that case a tour must be the boundary of
a convex polygon. This additional structural information then allows
them to deduce the existence of an optimal tour within
$\mathcal{O}(L^*)$ of a vertex of one of the polygonal
regions. Considering rectangles of the right size (since $L^*$ is
known up to a constant factor) yields the following lemma:

\begin{lemma} \label{lem:loc} 
  For a set $\mathcal{R}$ of disjoint, connected polygons in the plane
  with a total of $n$ vertices, $\mathcal{O}(n)$ rectangles of size
  $\mathcal{O}(L^*)$ can be found in polynomial time, such that an
  optimal tour of length $L^*$ is contained in at least one of them.
\end{lemma}

\subsection{Guillotine subdivisions and the charging scheme} \label{sec:guillotine}

If there were no further problems, a PTAS could be obtained by
applying \Mitchell's algorithm to all rectangles from
Lemma~\ref{lem:loc}. The main idea of the algorithm is to find a
nearly optimal tour that satisfies the $(m, M)$-guillotine property,
which will be defined in the following.

Consider a polygonal planar embedding $S$ of a graph $G$ with edge set
$E$ an a total length of $L$. Without loss of generality, let $E$ be a
subset of the interior of the unit square $B$. Let $\mathcal{R}$ be a
set of regions and $W_0$ the axis-aligned bounding box of an optimal
tour (we can afford to enumerate all possibilities on a grid and get
a $(1+\vare)$-approximation of $W_0$); let $\mathcal{R}_{W_0}$ be the
subset of regions that lie in the interior of $W_0$.

\begin{definition}[\cite{mitchell-ptas}]
  A \textbf{window} is an axis-aligned rectangle $W \subseteq B$. Let
  $l$ be a horizontal or vertical line through the interior of $W$,
  then $l$ is called a \textbf{cut} for $W$. 

  The intersection $l \cap E \cap \, \mathrm{int}(W)$ consists of a
  set of subsegments of the restriction of $E$ to $W$. Let $p_1,
  \dots, p_\xi$ be the endpoints of these segments ordered along
  $l$. Then the \textbf{$\boldsymbol{m}$-span} $\sigma_m(l)$ of $l$
  (with respect to $W$) is empty, if $\xi \leq 2m-2$, and consists of
  the line segment $\overline{p_m p_{\xi-m+1}}$ otherwise (see
  Figure~\ref{pic:cuts}). A cut $l$ is \textbf{$\boldsymbol{m}$-good}
  with respect to $W$ and $E$, if $\sigma_m(l) \subseteq E$.
\end{definition}

\begin{figure}[hbt] \center
  \includegraphics[width=12cm]{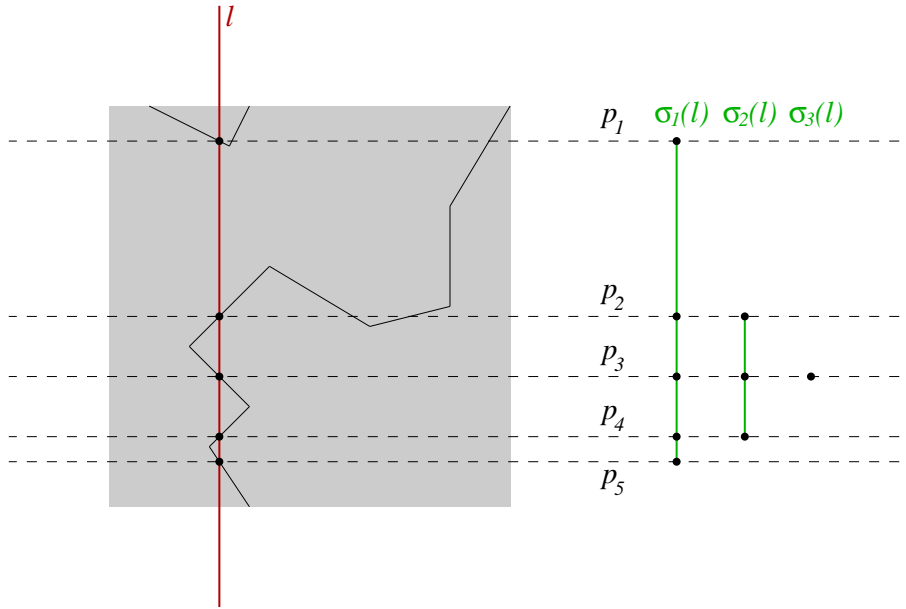}
  \caption{A cut $l$ and its $m$-span for $m=1,2,3$. The cut is
    $3$-good, but not $2$-good.}
  \label{pic:cuts}
\end{figure}

Mitchell defines the $M$-region-span analogously:

\begin{definition}[\cite{mitchell-ptas}] \label{def:m-region-span} The
  intersection $l \cap \mathcal{R}_{W_0} \cap \, \mathrm{int}(W)$ of a
  cut $l$ with the regions $\mathcal{R}_{W_0}$ restricted to $W$
  consists of a set of subsegments of $l$. The
  \textbf{$\boldsymbol{M}$-region-span} $\Sigma_M(l)$ of $l$ is the
  line segment $\overline{p_M p_{\xi-M+1}}$ along $l$ from the $M$th
  entry point $p_M$, where $l$ enters the $M$th region of
  $\mathcal{R}_{W_0}$, to the $M$th-from-the-last exit point
  $p_{\xi-M+1}$, assuming that the number of intersected regions is
  $\xi > 2(M-1)$. Otherwise, the $M$-region-span is empty.
\end{definition}

\hspace*{-0.7cm}
  \begin{minipage}{14.8cm}
    \quad This definition is ambiguous if the regions are not
    required to be convex, because the order of the regions is unclear
    and there might be a number of points at which $l$ enters or exits
    the same region. For example, on the right, many of the
    line segments connecting two red dots could be the $1$-region-span
    according to this definition.
  \end{minipage}
  \hfill
  \begin{minipage}{2cm} \centering
        \includegraphics[width=\linewidth]{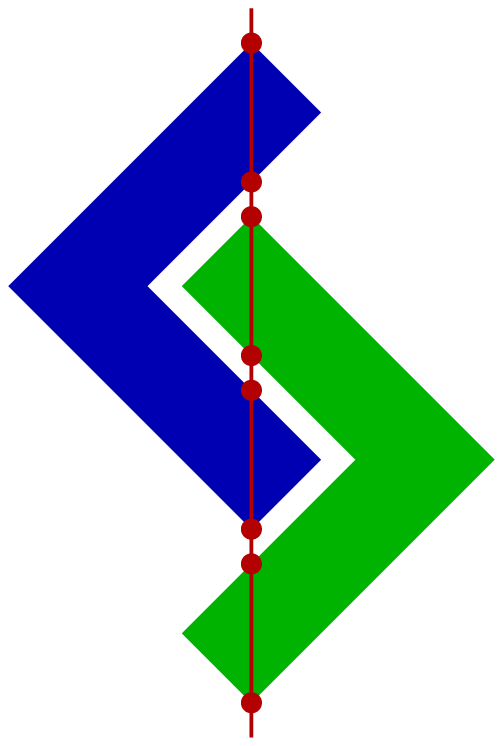}
  \end{minipage}

Furthermore, \Mitchell's $M$-region-span does not ``behave well'' in
the corresponding charging scheme. We propose the following
alternative definition. Its benefits will become apparent in
the proof of Theorem~\ref{thm:charge} and in
Figure~\ref{pic:charging-problems}:

\begin{definition} \label{def:region-span} The intersection $l \cap
  \mathcal{R}_{W_0} \cap \, \mathrm{int}(W)$ of a cut $l$ with the
  internal regions $\mathcal{R}_{W_0}$ restricted to $W$ consists of a
  (possibly empty) set of subsegments of $l$. Let $p_1, \dots, p_\xi$
  be the endpoints of these segments which are in $\mathrm{int}(W)$,
  ordered along $l$. Then the \textbf{$\boldsymbol{M}$-region-span}
  $\Sigma_M(l)$ of $l$ (with respect to $\mathcal{R}_{W_0}$ and $W$)
  is empty, if $\xi \leq 2M-2$ and consists of the line segment
  $\overline{p_M p_{\xi-M+1}}$ otherwise (see Figure~\ref{pic:rcuts}).

  A cut $l$ is \textbf{$\boldsymbol{M}$-region-good} with respect to
  $W$, $\mathcal{R}_{W_0}$ and $E$, if $\Sigma_M(l) \subseteq E$.
\end{definition}

\begin{figure}[hbt] \center
  \includegraphics[width=12cm]{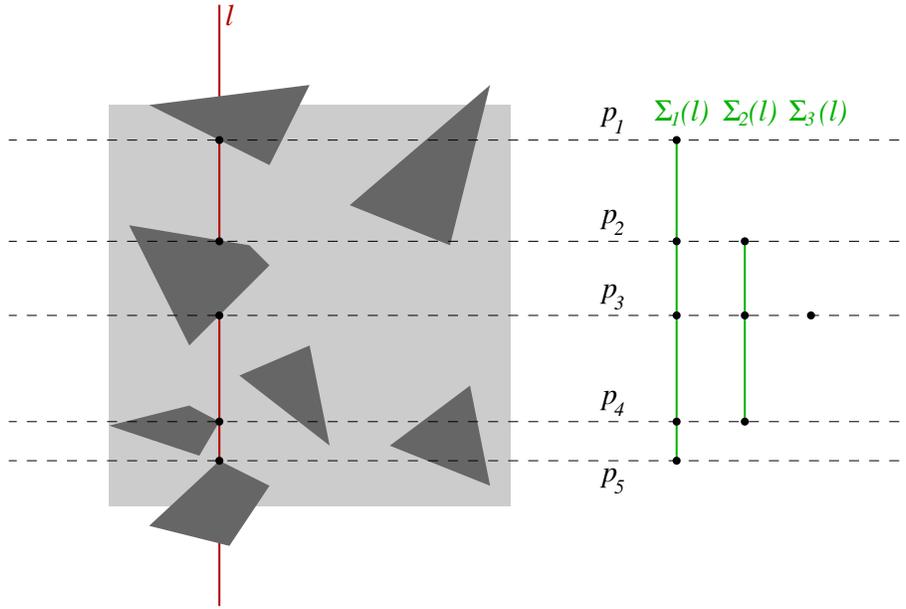}
  \caption{A cut $l$ and its $M$-region-span (according to
    Definition~\ref{def:region-span}) for $M=1,2,3$.}
  \label{pic:rcuts}
\end{figure}

With either definition of the $M$-region-span, we can define a
corresponding version of the $(m, M)$-guillotine property as follows:

\begin{definition}[{\cite{mitchell-ptas}}] 
  An edge set $E$ of a polygonal planar embedded graph satisfies the
  \textbf{$\boldsymbol{(m, M)}$-guillotine property} with respect to a window $W$
  and regions $\mathcal{R}_{W_0}$, if one of the following conditions
  holds:
  \begin{itemize}
    \item No edge of $E$ lies completely in the interior of $W$, \textit{or}
    \item There is a cut $l$ of $W$ that is $m$-good with respect to
      $W$ and $E$ and $M$-region-good with respect to $W$,
      $\mathcal{R}_{W_0}$ and $E$, such that $l$ splits $W$ into two
      windows $W'$ and $W''$, for which $E$ recursively satisfies the
      $(m, M)$-guillotine property with respect to $W'$ resp. $W''$
      and $\mathcal{R}_{W_0}$. 
  \end{itemize}
\end{definition}

It is clear from this definition that transforming a tour into an edge
set with this property will induce an additional length that depends
both on the edges and the regions present. The crucial property of a
tour connecting disjoint, $\alpha$-fat regions is that their number
and diameter provide a lower bound on its length. It is worth noting
that the following lemma holds for a tour among $\alpha$-fat regions
in either Mitchell's (Definition~\ref{def:afat-mit}) or Elbassioni's
and Fishkin's (Definition~\ref{def:afat-elb}) sense:

\begin{lemma}[{\cite[Lemma 2.6]{mitchell-ptas}}] \label{lem:afat-mit}
  Let $\varepsilon > 0$, then there is a constant $C$ (that depends on
  $\vare$ and $\alpha$), such that for every \textsc{TSPN}-tour $T^*$
  of length $L^*$, connecting $k$ disjoint, connected, $\alpha$-fat
  ($\alpha$-fat$_E$) regions in the plane, $L^* \geq C \cdot
  \dfrac{\lambda(\mathcal{R}_{W_0})}{\log (\frac{k}{\varepsilon})}$,
  where $\lambda(\mathcal{R}_{W_0})$ is the sum of the diameters of
  the regions that are completely contained in the axis-aligned
  bounding box $W_0$ of $T^*$.
\end{lemma}

\cite{mitchell-ptas} provides a proof for this lemma with respect to
$\alpha$-fat regions in the sense of Definition~\ref{def:afat-mit},
which can easily be adapted for $\alpha$-fat$_E$ regions as in
Definition~\ref{def:afat-elb} (even without requiring connected
regions).

In the dynamic programming algorithm, $M$ can be chosen as
$\mathcal{O}(\frac{1}{\varepsilon} \log (\frac{k}{\varepsilon}))$;
therefore we can ``afford'' to construct additional edges of length
$\mathcal{O}(\frac{\diam(P_i)}{M})$ for every $P_i \in
\mathcal{R}_{W_0}$ and still obtain a $(1 +
\mathcal{O}(\varepsilon))$-approximation algorithm.

The following definitions were not explicitly given in
\cite{mitchell-ptas} and are therefore adapted from the corresponding
definitions in \cite{mitchell-tsp} for the standard TSP:

\begin{definition}
  Let $l$ be a cut through window $W$, and $p$ a point on $l$, then
  $p$ is called \textbf{$\boldsymbol{m}$-dark} with respect to $W$ and
  an edge set $E$, if $p$ is contained in the $m$-span of the cut
  through $p$ that is orthogonal to $l$.

  Similarly, a point $p$ on a cut $l$ is said to be
  \textbf{$\boldsymbol{M}$-region-dark}, if it is contained in the
  $M$-region-span of a cut through $p$ that is orthogonal to $l$.
  
  A segment on a cut $l$ is called \textbf{$\boldsymbol{m}$-dark} and
  \textbf{$\boldsymbol{M}$-region-dark}, respectively, if every point
  of it is.

  A cut $l$ is called \textbf{favorable} if the sum of the lengths of
  its $m$-dark and $M$-region-dark portions is at least as big as the
  sum of the lengths of its $m$-span and $M$-region-span.
\end{definition}

While our definition of the $M$-region-span removes the ambiguity and
ensures the correctness of the proof techniques used by Mitchell, it
yields a weaker (but correct) overall statement:

\begin{theorem}[{Corrected version of \cite[Theorem
    3.1]{mitchell-ptas}}] \label{thm:charge} Let $G$ be an planar
  embedded connected graph, with edge set $E$ consisting of line
  segments of total length $L$. Let $\mathcal{R}$ be a set of
  disjoint, polygonal, $\alpha$-fat regions and assume that $E \cap
  P_i \neq \emptyset$ for every $P_i \in \mathcal{R}$. Let $W_0$ be
  the axis-aligned bounding box of $E$. Then, for any positive
  integers $m$ and $M$, there exists an edge set $E' \supseteq E$ that
  obeys the $(m, M)$-guillotine property with respect to $W_0$ and
  regions $\mathcal{R}_{W_0}$, and for which the length of $E'$ is at
  most $L + \frac{\sqrt{2}}{m} L + \frac{\sqrt{2}}{M}
  \Lambda(\mathcal{R}_{W_0})$, where $\Lambda(\mathcal{R}_{W_0})$ is
  the sum of the perimeters of the regions in $\mathcal{R}_{W_0}$.
\end{theorem}

The only deviation from \Mitchell's version is that the length of $E'$
is bounded using $\Lambda(\mathcal{R}_{W_0})$ instead of
$\lambda(\mathcal{R}_{W_0})$ as defined in Lemma~\ref{lem:afat-mit}.

To apply the lower bound on the optimum obtained from $\alpha$-fatness
as in the original paper, a further restriction can be imposed on the
regions -- that the ratio between the perimeter and diameter of
regions is bounded by a constant, which is true, for example, for
convex regions. Since polygonal regions only ensure that this ratio is
bounded by $\mathcal{O}(n)$, this is a very restrictive assumption.

Note further that this theorem, as well as \cite[Theorem
3.1]{mitchell-ptas}, does not establish the existence of a
\emph{connected} edge set with the properties of $E'$; see
Section~\ref{sec:cnn}.

The proof relies on the following key lemma by Mitchell: 

\begin{lemma}[{\cite[Lemma 3.1]{mitchell-ptas}}] \label{lem:fav-cut}
  For any planar embedded graph $G$ with edge set $E$, any set of
  regions $\mathcal{R}_{W_0}$ and any window $W$, there is a favorable
  cut.
\end{lemma}

Now, given an edge set $E$ as in the theorem, we recursively find a
favorable cut $l$, add the $m$-span and $M$-region-span to $E$ and
proceed with the two subwindows, into which $l$ splits the current
window.

This procedure terminates, because the proof of
Lemma~\ref{lem:wfav-cut} yields that the cut can always be chosen to be
at one of finitely many candidate coordinates since we assume that
all vertices of the tour lie on a grid. 

As for the additional length induced by the $m$-spans and
$M$-region-spans, we know that it can be bounded by the length of the
respective dark portions of the cuts in question. Since
\cite{mitchell-ptas} omits some of the details, we will give them
here.

\begin{proof}[Theorem~\ref{thm:charge}]

  The charging scheme works as follows: Every edge and the boundary
  (in \Mitchell's version, diameter) of every region is split up into
  finitely many pieces, to each of which we assign a ``charge'' that
  specifies which multiple of the length of that segment was added to
  $E$ as part of $m$-spans and $M$-region-spans. If we can establish
  that the charge for every edge segment is at most
  $\frac{\sqrt{2}}{m}$, the charge for every region boundary is at
  most $\frac{\sqrt{2}}{M}$, and the $m$-span and $M$-region-span
  never get charged during the recursive process, we obtain the
  statement of Theorem~\ref{thm:charge}.

Let $l$ be a favorable cut. The charging scheme for the edge set is
described in \cite{mitchell-tsp}: For each $m$-dark portion of $l$,
the $2m$ inner edge segments (the $m$ segments closest to the cut on
each side) are each charged with $\frac{1}{m}$. In the recursive
procedure, each segment $e$ can be charged no more than once from each
of the four sides of its axis-parallel bounding box, since in order
for it to be charged, there have to be at least $m$ edges to the
corresponding side of it, but there are less than $m$ edges between
$e$ and any cut that charges it. Therefore, after placing a cut and
charging $e$ from one side, there will be less than $m$ edges to the
respective side of $e$ in the new subwindow, preventing it from being
charged from that direction again.

Thus, each side of the axis-parallel bounding box of the segment gets
charged $\frac{1}{2m}$ times, and since the perimeter of the bounding
box is at most $2 \sqrt{2}$ times the length of the edge segment $e$,
it gets charged no more than $\frac{\sqrt{2}}{m}$ times in total.

The $m$-span and $M$-region span never get charged, because after
they are inserted, they are not in the interior of any of the windows
which are considered afterwards. 

With Definition~\ref{def:region-span} of the $M$-region-span, it is
possible to replace the regions by their boundaries (which form a
polygonal edge set of total length $\Lambda(\mathcal{R}_{W_0})$) and
to treat them the same way as the edge set $E$ (in particular, the
$M$-region-span and $M$-region-dark parts become $M$-span and
$M$-dark).
\end{proof}

\begin{figure}[hbt] \center
  \includegraphics[width=6cm]{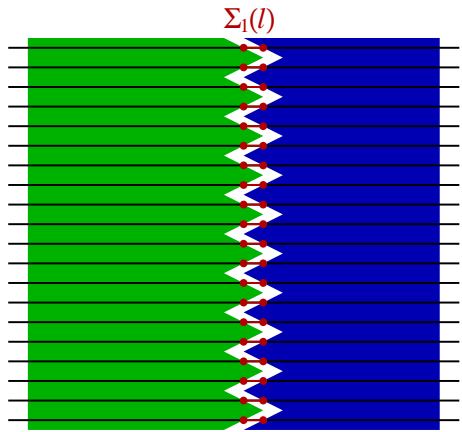}
  \caption{Charge is proportional to perimeter}
  \label{pic:charging-problems}
\end{figure}

For \Mitchell's original definition
(Definition~\ref{def:m-region-span}) of the $M$-region-span, a
scenario as in Figure~\ref{pic:charging-problems} can become a
problem. Every black line pictured is a favorable cut, every red line
segment is $1$-region-dark on the respective cut (even if the window
in question has been cut by the black line directly below and above
already). The total length of the red line segments is however
proportional to the perimeter, not the diameter, of the blue and green
regions.

Within one window, \Mitchell's statement holds; the problem with his
definition is its lack of a monotone additive behavior: When cutting a
window $W$ into two parts, the sum of the diameters of all relevant
regions is $W$ might be less than the sums of the diameters of the
relevant regions for each subproblem combined, and while no part of
the diameter of a region is charged more than $\frac{\sqrt{2}}{M}$
times in each subproblem, the combined charge might still be greater
than $\frac{\sqrt{2}}{M}$.

\subsection{Grids and guillotines} \label{sec:grid}

In order for the dynamic programming algorithm to work, the number of
possible endpoints for an edge has to be restricted (for example, to a
grid). In \cite{mitchell-ptas}, an optimal solution will thus first be
moved to a fine grid through slight perturbation, and subsequently
transformed into an $(m, M)$-guillotine subdivision. Mitchell claims
that there is always a favorable cut that has grid coordinates,
arguing that in the charging lemma (Lemma~\ref{lem:wfav-cut}), the
functions considered are piecewise linear between grid points,
therefore the maximum of such a function must be attained at a grid
point. The proof fails to take into account that the function might be
discontinuous (and not even semi-continuous) at grid points.

Even if this were true, it is not in general true in the Euclidean
case (unlike the rectilinear case) that the $m$-span ends at a grid
coordinate on the cut (for example, it could instead end at an
interior point of an edge).

A (slightly technical) solution to this uses a weaker version of this
claim, i. e. that a favorable cut has to be at a grid coordinate or
the mean value of two consecutive grid points, which follows from the
simple observation that when integrating affine functions over an
interval, the sign of the integral is the same as the sign of the
function value at the midpoint of the interval.

This can be used in the proof of Lemma~\ref{lem:fav-cut} as follows:
The existence of a favorable cut is shown via changing the order of
integration -- then the integral of the length of the dark portions
along the $x$-axis is the same as the integral of the length of the
spans along the $y$-axis, and vice versa.

Therefore, there is one axis, such that there is more dark than
spanned area in that direction, i. e. the total area of all dark
points with respect to some horizontal (resp. vertical) cut is greater
than the area of all points that are contained in some horizontal
(resp. vertical) cut. Thus, there has to be a single cut with that
property as well: a favorable cut.

If all previous edges and regions are restricted to the grid, the
aforementioned observations imply that in particular, there is a
favorable cut at a grid coordinate or in the center between two
consecutive ones. 

It can then be shown that a non-empty $m$-span in an optimal solution
always has a certain minimum length, or it contains a grid point. This
observation allows us to slightly modify the edge set, so that a
cut becomes $m$-good, while all edges still have grid endpoints.

The $M$-region-span can be dealt with in a similar way. Moving it to
the grid requires a slight change in the definition of the $(m,
M)$-guillotine property, which will preserve its algorithmic
properties. However, there are further problems with the
$M$-region-span, which will be explained in the following section.

\subsection{Connectivity} \label{sec:cnn} 

To transform an optimal tour into an $(m, M)$-guillotine subdivision,
the $m$-span and $M$-region-span of a favorable cut are inserted into
the edge set through a recursive procedure. The $m$-span is always
connected to the original edge set $E$, since its endpoints are
intersection points of the cut with $E$. This is not true for the
$M$-region-span, and in fact, it can be ``far away'' from $E$, as seen
in Figure~\ref{pic:region-span}. The optimal tour (green) and the
$1$-region-span $\Sigma_1(l)$ of the favorable cut $l$ (which is
favorable, because the two grey squares at the top make a portion of
$l$ with the same length as $\Sigma_1(l)$ $1$-region-dark) are far
away from each other. Connecting it to the tour does not preserve the
approximation ratio of $1+\vare$. 

\begin{figure}[hbt]
  \centering
  \includegraphics[width=3cm]{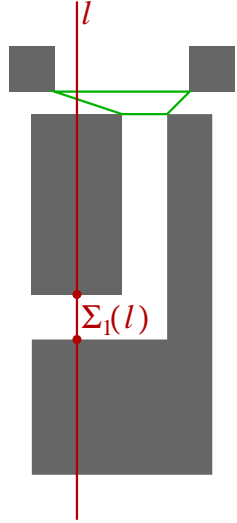}
  \caption{Favorable cut and disconnected $1$-region-span}
  \label{pic:region-span}
\end{figure}

Note that, in the dynamic programming algorithm, we cannot afford to
decide whether to connect the $M$-region-span of a cut to the edge
set: If we choose not to connect it, we have to decide which
subproblems is responsible for each region on the $M$-region-span, but
this is exactly what was to be avoided by introducing it in the first
place. 

On the other hand, if we do connect the $M$-region-span to the edge
set, both its length and its possible interference with other
subproblems have to be taken care of. With the second definition of
$\alpha$-fatness (Definition~\ref{def:afat-elb}), which implies the
lower bounds mentioned in Section~\ref{sec:ext}, the length of a
segment connecting the $M$-region-span to $E$ could be charged off to
the length of the $M$-region-span itself, whereas Mitchell's
definition of $\alpha$-fatness does not even guarantee this (because
in the proof of the lower bound, we relied on the regions being
contained in the bounding box of the tour). It is not clear whether
the connection of tour and region-span intersects another subproblem,
possibly violating the $(m, M)$-guillotine property there, therefore
even for $\alpha$-fat$_E$ regions, this problem remains open.

\subsection{External regions} \label{sec:ext}

In addition to dealing with the internal regions $\mathcal{R}_{W_0}$,
the dynamic programming algorithm has to determine how to visit external
regions. \Mitchell's strategy is to enumerate all possible options,
restricting the complexity with the following argument: Given a
situation as in Figure~\ref{pic:sub}, with some external regions
protruding from the outside into a subproblem $W$, we know that since
there are only $\mathcal{O}(m)$ edges on each side of $W$, they can be
split into $\mathcal{O}(m)$ intervals of regions, such that along the
corresponding side of $W$, the regions are consecutive with no edges
passing through $\partial W$ between them (e. g. the red, green and
blue region in Figure~\ref{pic:sub}). 

It seems clear now that, in order for the green region to be visited
by an edge outside of $W$, either the red or the blue region would
have to be crossed as well. If this were true, it could be deduced
that the set of regions in one of the intervals in question that are
not visited outside $W$, and that thus $W$ is responsible for, is a
connected subinterval, leading to $\mathcal{O}(n^2)$ possibilities for
each interval and $\mathcal{O}(n)^{\mathcal{O}(m)}$ possibilities
overall for each window $W$. 

This argument fails if the green region has a disconnected
intersection with $W_0$. An example is given on the left in Figure~\ref{pic:ext}:
An $(m, M)$-guillotine tour can visit any subset of the regions
outside of $W$, thus there is no polynomial upper bound
on the number of possibilities anymore.

Mitchell's argument holds for convex regions, but as seen left in
Figure~\ref{pic:ext}, in general it does not apply to disjoint,
connected, $\alpha$-fat regions. The number of regions such that their
intersection with $W_0$ (or even a slightly extended rectangle) is
disconnected could be $\Theta (k)$, for example if the construction in
Figure~\ref{pic:ext} is extended beyond the yellow region, which is
possible, because the regions here actually become ``more fat'' as
their size increases, i.e. $\alpha$ decreases and eventually converges
to 2.

\begin{figure}[hbt]
  \centering
  \includegraphics[width=3cm]{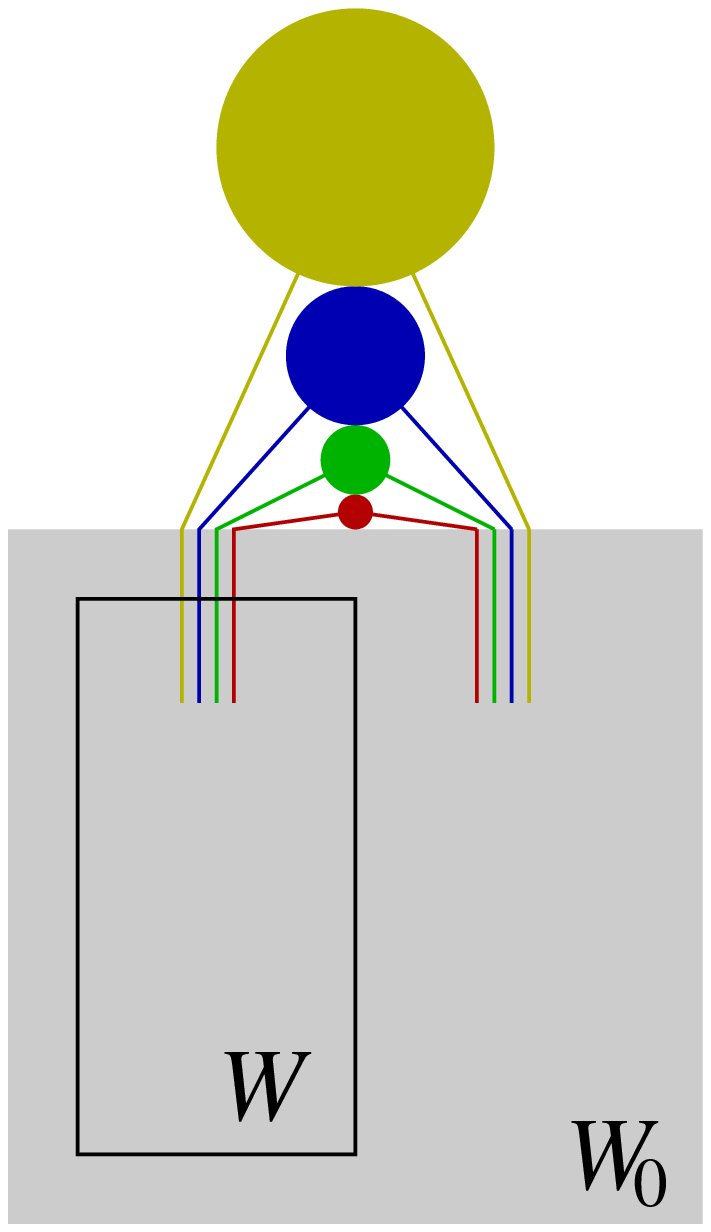}
  \hspace*{0.3cm}
  \includegraphics[width=5cm]{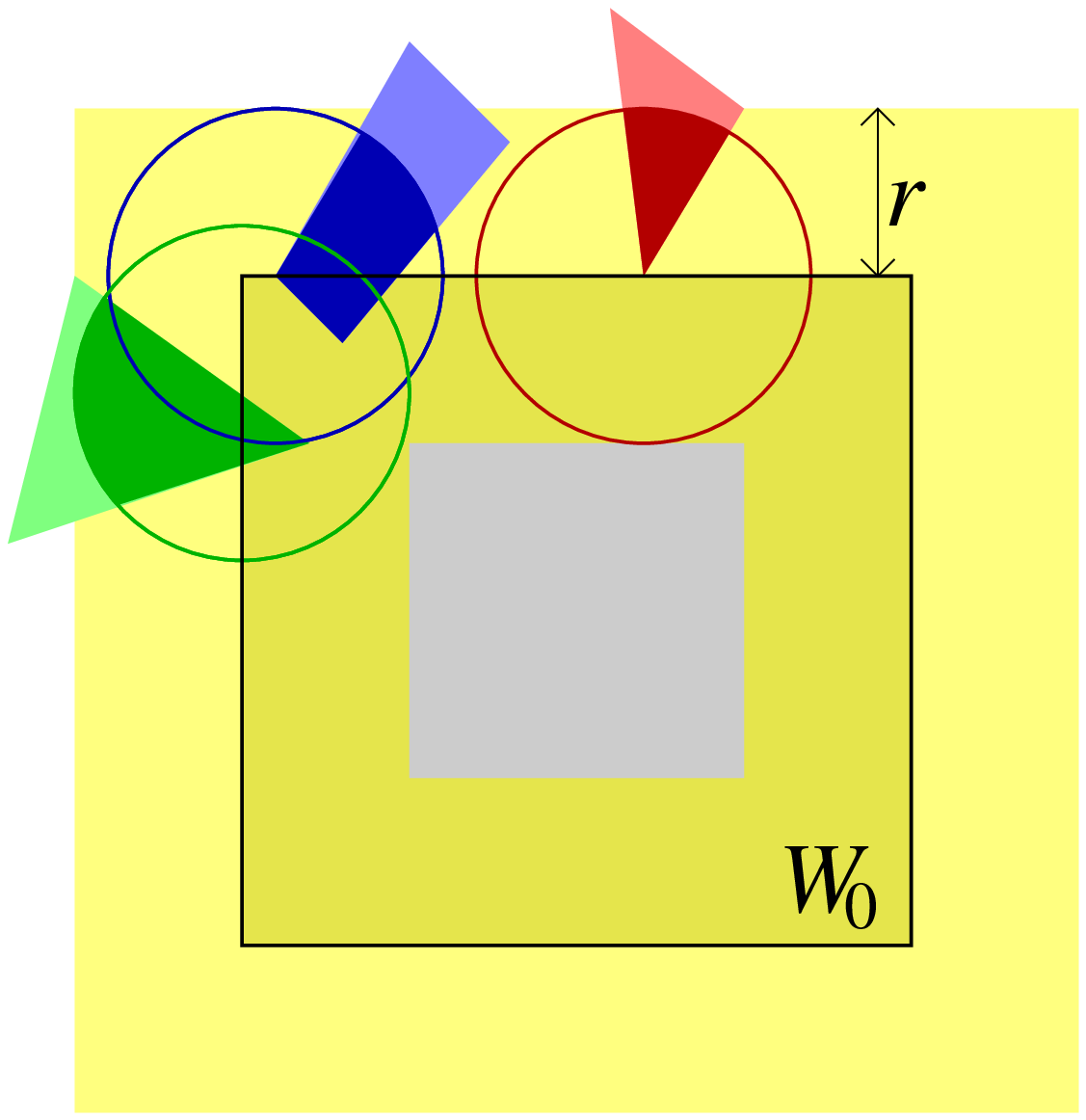}
  \caption{External regions}
  \label{pic:ext}
\end{figure}

It can be shown that in order for this to be a problem, the size of
the regions has to increase exponentially due to the logarithmic lower
bound in the packing lemma, and the fact that the boundary of $W_0$ is
a tour of the external regions. One solution is therefore to restrict
the diameter of the regions; many of the approximation algorithms for
similar problems do in fact require the regions to have comparable
diameter (\cite{elb2}, \cite{dumi}).

Alternatively, requiring convexity solves the issue, but is quite a
strong condition. Another option is using the notion of
$\alpha$-fatness$_{E}$ from Definition~\ref{def:afat-elb} as
established by K. Elbassioni, A. Fishkin, N. Mustafa and R. Sitters
\cite{elb}.

This definition implies that a path connecting $k$ regions has a
length of at least $(\frac{k}{\alpha}-1) \frac{\pi \delta}{4}$, where
$\delta$ is the diameter of the smallest region \cite{elb}; adding a
variant of this up by diameter types yields the same lower bound as
for \Mitchell's definition of $\alpha$-fatness, up to a constant
factor.

Unlike \Mitchell's version, this definition however estimates the
length of a tour in terms of the minimum diameter of the regions
involved and can therefore be used to give a constant upper bound on
the number of large external regions (see Figure~\ref{pic:ext}, on the
right): since $\partial W_0$ is a path connecting them, the number of
external regions with diameter $\geq \diam(W_0)$ is at most
$\alpha(\frac{8\sqrt{2}}{\pi}+1)$. 

In both cases, the small external regions can be added to
$\mathcal{R}_{W_0}$, since $\partial W_0$ is a tour of them, which is
at least $\frac{1}{\sqrt{2}}$ times the length of an optimal tour, and
thus these regions provide a lower bound of the length of $\partial
W_0$, which in turn provides a lower bound on $L^*$. This way, the
statement of Lemma~\ref{lem:afat-mit} (for which the fact that $W_0$
is the bounding box of the tour was exploited in the proof) remains
intact with modified constants. For the large regions, we can afford
to explicitly enumerate which subwindow should visit them. 

\subsection{Result}
Overall, we have the following result:
\begin{theorem} \label{thm:main} Let $\vare > 0$ be fixed. Given a set
  of $k$ disjoint, connected, polygonal regions that are
  $\alpha$-fat$_E$, convex or $\alpha$-fat and of bounded diameter,
  with a total of $n$ vertices in the plane, we can find a connected,
  $(m, M)$-guillotine, Eulerian grid-rounded graph visiting all
  regions in polynomial time in the size of the grid, $n$, $k$, $2^M$
  and $(nm)^m$. Among all such graphs, it will be shortest possible up
  to a factor of $1 + \vare$.
\end{theorem}

Here, grid-rounded means that all edge endpoints are on a grid of
polynomial size, and that there is only a polynomial number of
possible positions for every cut. 

Mitchell claims that for $M = \mathcal{O}(\frac{1}{\vare} \log
\frac{n}{\vare})$ and $m = \mathcal{O}(\frac{1}{\vare})$ and
$\alpha$-fat regions, a connected $(m,M)$-guillotine
graph is a $(1+\vare)$-approximation of a tour; his proofs only apply
to not necessarily connected graphs and regions with bounded
perimeter-to-diameter ratio. In general, because of the connectivity
problem in Section~\ref{sec:cnn} and some technical difficulties
choosing an appropriate grid, it is not clear whether there is a grid
of polynomial size, such that a graph with the properties of the
theorem is a $(1+\vare)$-approximation of an optimal TSPN tour. For
unit disks and with a slightly modified definition of the guillotine
property, there are $m$ and $M$ such that the guillotine subdivision
is only by a factor of $(1+ \mathcal{O}(\vare))$ longer than a tour;
this will be shown in Theorem~\ref{thm:grid}.

\section{Unit Disks}
The criticism put forward in Section~\ref{sec:cnn} extends to a joint
paper of {\Mitchell} and Dumitrescu \cite{dumi}. Despite being
published earlier than the PTAS candidate, the approach chosen there
actually takes into account that the $M$-region-span (or $m$-disk-span
in the notation of the paper) has to be connected to the edge set, and
this is done at a sufficiently low cost. However, no proof is given
that the edges added during this process preserve the $(m,
M)$-guillotine property. In particular, even if a subproblem contains
disks that do not intersect its boundary, the $M$-region-span might
not visit one of them; if it always did, then we could add the
connecting edge and guarantee that it remains within the same
subproblem. 

With some additional effort, this problem can be avoided, as we will
show now. 

\subsection{Preliminaries}
Definition~\ref{def:afat-elb} yields a useful lower bound:

\begin{lemma}[\cite{elb}] \label{lem:afat} A shortest path connecting
  $k$ disjoint, $\alpha$-fat\E regions of diameter $\geq \delta$ has
  length at least $(\frac{k}{\alpha}-1) \frac{\pi \delta}{4}$.
\end{lemma}

All results apply not only to unit disks (for which $\alpha = 4$), but
to disk-like regions:

\begin{definition}
  A set of regions is \textbf{disk-like}, if all regions are disjoint
  and connected, and have comparable size (their diameters range
  between $d_1$ and $d_2$, which are constant), are $\alpha$-fat\E or
  $\alpha$-fat for constant $\alpha$, and their perimeter-to-diameter
  ratio is bounded by a constant $r$.
\end{definition}

\subsection{Charging Scheme} \label{sec:charge}

Let $\mathcal{D}$ denote the input set of $k$ disjoint unit
disks. 

Throughout the rest of this paper, we will use a slightly modified
version of the $(m, M)$-guillotine property:

\begin{definition}[{\cite{mitchell-ptas}}] 
  An edge set $E$ of a polygonal planar embedded graph satisfies the
  \textbf{$\boldsymbol{(m, M)}$-guillotine property} with respect to a window $W$
  and regions $\mathcal{R}_{W_0}$, if one of the following conditions
  holds:
  \begin{enumerate}
  \item There is no edge in $E$ with its interior (i. e. the edge
    without its endpoints) completely contained in the interior of
    $W$.
  \item There is a cut $l$ of $W$ that is $m$-good with respect to $W$
    and $E$ and $M$-region-good with respect to $W$,
    $\mathcal{R}_{W_0}$ and $E$, such that $l$ splits $W$ into two
    windows $W'$ and $W''$, for which $E$ recursively satisfies the
    $(m, M)$-guillotine property with respect to $W'$ resp. $W''$ and
    $\mathcal{R}_{W_0}$.
  \end{enumerate}
\end{definition}
  
The first case differs from Mitchell's definition, which only requires
that no entire edge lies completely in the interior of $W$. However,
with that definition, adding the $m$-span to the edge set does not
reduce the complexity of the subproblem (possibilities for edge
configurations on the boundary of the window), because then we would
still have to know the positions of the edges that intersect the
$m$-span.

\begin{definition}
  Let $m$ and $M$ be fixed. Then a cut is called
  \textbf{$\boldsymbol{c}$-favorable}, if the sum of the lengths of
  its $m$-span and $M$-region-span is at most $c$ times the sum of the
  lengths of its $m$-dark and $M$-region-dark portions.

  A cut is \textbf{weakly $\boldsymbol{c}$-favorable}, if the sum of
  the lengths of its $m$-span and $M$-region-span is at most $c$ times
  the sum of the lengths of its $m$-dark and $M/2$-region-dark
  portions.
\end{definition}

For a cut $l$, define the following notation: 
\begin{itemize}
\item $\sigma_m(l)$ -- length of the $m$-span
\item $\Sigma_M(l)$ -- length of the $M$-region-span 
\item $\delta_m(l)$ -- length of the $m$-dark segments
\item $\Delta_M(l)$ -- length of the $M$-region-dark segments
\end{itemize}

\begin{definition}
Let $m$ and $M$ be fixed. A cut is called 
\textbf{weakly central}, if it is horizontal and has distance at
least $\min\{2, h/4\}$ from the top and bottom edge of the window or it
is vertical and has distance at least $\min\{2, w/4\}$ from the left and
right edge of the window, where $w$ and $h$ denote the width and
height of the window, respectively. 

It is \textbf{perfect}, if it is
weakly $8$-favorable and at least one of the following holds:
\begin{itemize}
\item it is \textbf{central}, i. e. if it is a horizontal cut, it has distance
at least 2 from its top and bottom edge; if it is a vertical cut, it
has distance at least 2 from the left and right edge, \textit{or}
\item it is weakly central and its $M$-region-span is empty.

\end{itemize}
\end{definition}

Any central cut is weakly central, any $c$-favorable cut is weakly
$c$-favorable. 

\begin{lemma} \label{lem:wfav-cut}
Let $m$ and $M\geq 24$ be fixed. Every window has a perfect cut. 
\end{lemma}

\begin{proof}
  Lemma 3.1 in \cite{mitchell-ptas} states that every window has a
  favorable cut. We can use the same techniques to show that almost
  every window has a central $2$-favorable cut: By definition, any
  point $p$ is in the $m$-span of a vertical cut, if and only if it is
  $m$-dark in a horizontal cut; analogously for regions. Therefore,
  $\int_x \sigma_m(l_x) + \Sigma_M(l_x) \dd x = \int_y \delta_m(l_y) +
  \Delta_M(l_y) \dd y$, where $l_x$ is the vertical cut with
  $x$-coordinate $x$ and $l_y$ is the horizontal cut with
  $y$-coordinate $y$; {\Wlog} $\int_x \sigma_m(l_x) + \Sigma_M(l_x)
  \dd x \leq \int_x \delta_m(l_x) + \Delta_M(l_x) \dd x$. Then there
  is a $1$-favorable vertical cut, i. e. an $x$ such that
  $\sigma_m(l_x) + \Sigma_M(l_x) \leq \delta_m(l_x) +
  \Delta_M(l_x)$. Using Markov's inequality, we can also conclude that
  at least half of the vertical cuts are $2$-favorable. Therefore, if
  the window has width $\geq 8$, we can choose a central $2$-favorable
  cut.

  If the window has width $a < 8$, then the same argument yields that
  there is still a $2$-favorable cut $l_x$ with distance at least
  $a/4$ from the left and right edge of the window. If its
  $M$-region-span is empty, it is a perfect cut. Otherwise, there are
  at least $2M$ intersection points with disks, i. e. at least $M$
  disks, each of which has to intersect this cut and thus have at
  least $a/4$ of its width within the window.

  Consider the interval $[x-a/8, x+a/8]$ of the window (and note that
  it has width $< 2$). Let $\mathcal{D}_x$ be the set of disks on
  $l_x$, then each of them must intersect $l_{x-a/8}$ or $l_{x+a/8}$
  (or both). Therefore, at least $M/2$ disks of $\mathcal{D}_x$
  intersect one of them, \Wlog, $l_{x-a/8}$. This means that every cut
  with $x$-coordinate in $(x-a/8, x)$ has a total of $M$ intersection
  points with these disks.  Let $y$ be the $y$-coordinate of the
  horizontal cut such that half of these intersection points are below
  and half of them above $l_y$. Then, the segment from $x-a/8$ to $x$
  on $l_y$ is $M/2$-region-dark.

  On the other hand, since $a < 8$, any horizontal cut can intersect
  at most $k < 4(\frac{16}{\pi} + 1) < 25$ disks (because
  Lemma~\ref{lem:afat} implies $a \geq (\frac{k}{4} - 1)
  \frac{\pi}{2}$). Therefore, there are at most $48$ total
  intersection points with disks on the cut. As $M \geq 24$, $l_y$ has
  an empty $M$-region-span.

  For $l_y$, we now have $\sigma_m(l_y) \leq a$, $\Sigma_M(l_y) = 0$,
  $\delta_m(l_y) \geq 0$, and $\Delta_{M/2}(l_y) \geq a/8$. This
  implies that $l_y$ is weakly $8$-favorable.

  Finally, since there are at least $12$ disks above and below $l_y$
  on $l_x$, $l_y$ has distance $4 \pi/4 > 2$ from the top and
  bottom of the window, so it is weakly central and therefore a
  perfect cut. 
\end{proof} 

\begin{theorem}[Connected guillotines] \label{thm:guillotines} Let $m
\geq 32$ and $M \geq 24$ be fixed, and let
$\mathcal{R}$ be a set of $k \geq 20$ unit disks. Let $L^*$ be the
length of a shortest tour with edge set $E^*$ connecting them and
$W_0$ its axis-parallel bounding box, then there exists a connected
Eulerian $(m,M+24)$-guillotine subdivision with edge set $E' \cup E^*$
of length $(1 + \mathcal{O}(1/m) + \mathcal{O}(1/M)) L^*$ connecting
all regions.
\end{theorem}

Using different constants for $m$ and $M$ is not necessary here. Note,
however, that the algorithm has polynomial running time if $m \in
\mathcal{O}(1)$ and $M \in \mathcal{O}(\log n)$, so choosing
$M$ differently might help with different applications. 

\begin{proof}[Theorem~\ref{thm:guillotines}] 
  We recursively partition $W_0$ using perfect cuts. These always
  exist by the previous lemma. For each such cut, we add its $m$-span
  and $(M+24)$-region-span to the edge set (and not the
  $M$-region-span, because if a segment is added, we need a lower
  bound on the length of the $M$-region-span) as well as possibly an
  additional segment for connectivity (see Figure~\ref{pic:guill}).

  \begin{figure}[hbt] \center
    \includegraphics[width=7cm]{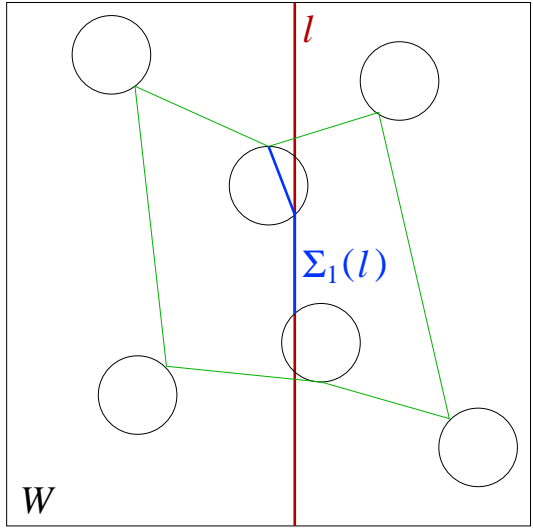} 
    \caption{Adding the blue $1$-region-span $\Sigma_1(l)$ and the
      connecting segment (twice) makes $l$ $1$-region-good}
    \label{pic:guill}
  \end{figure}

We refine Mitchell's charging scheme and assign to each point $x$ of
an edge or the boundary of a disk a ``charge" $c(x)$, such that the
additional length incurred throughout the construction equals $\sum_{D
  \in \mathcal{D}} \int_{\partial D} c(x) \dd x + \sum_{e \in E}
\int_{e} c(x) \dd x$. This charge will be piecewise constant. We will
show that the charge on an edge segment is bounded by $C'/m$, and the
charge of a disk boundary segment is bounded by $C/M$, for constants
$C'$ and $C$. This proves the theorem, because $\sum_{D \in
  \mathcal{D}} \int_{\partial D} C/M \dd x = C \cdot k 2\pi/M$, but
$L^* \geq (k/4 - 1) \cdot 2\pi/4 \geq 2\pi$, hence
$$C \cdot k 2\pi/M \leq \dfrac{8
\pi}{M} \left(\dfrac{4L^*}{\pi} + 1\right) = \dfrac{32C}{M} L^* +
\dfrac{8C\pi}{M} \leq \dfrac{36C}{M} L^*$$ for the disks, and for the
edges, $\sum_{e \in E^*} \int_{e} C'/m \dd x = \frac{C'}{m} L^*$.

From now on, a segment will refer to a disk boundary or edge
segment. In the beginning, every segment has charge 0. Every charge
that is applied to a segment gets added to its previous charge.

We will distinguish direct and indirect charge, and show that each
segment is directly charged at most 4 times, once from each
axis-parallel direction. Indirect charge will be charge that is added
to a segment in $E'$ that was constructed during the proof. Since we
cannot charge these segments, we pass their charge on: The new
segments at some point were charged to a segment in $E^*$ or $\bigcup_{D \in
\mathcal{D}} \partial D$, to which we add the new charge
recursively. For example, in Figure~\ref{pic:guill}, the blue region
span cannot be charged by any cut, since it is on the boundary of a
window. On the other hand, the connecting segment might be charged by
a different cut during the construction. If the direct charge for
inserting the blue edges was applied to the disks making (different)
parts of the cut 1-region-dark, then if the connecting segment is
charged, we will instead pass the charge on to the disks (each of them
receives half the charge). 

Now, let $W$ be a window with perfect cut $l$. Then,
we add its $m$-span and $(M+24)$-region-span to $E'$. Furthermore, we
have that if the $(M+24)$-region-span is non-empty, it is not
necessarily connected to the tour. But in this case the cut is
central. Therefore, no disk intersected by the $(M+24)$-region-span
can intersect the boundary of the window ({\Wlog} let $l$ be vertical:
Then no disk can intersect the left and right boundary, because the
cut is central, and the parts of the cut above and below the
region-span intersect at least 12 disks each, therefore their length
is at least 1). Connect the $(M+24)$-region-span to the closest point
of the tour within the same window (which will have distance $\leq 2$,
because all the disks are visited). Note the the $m$-span, by
definition, is connected to the edge set, hence the connecting
segments for the $M$-span are sufficient for connectivity.

This makes the cut $m$-good and $(M+24)$-region-good and leaves the
edge set connected, therefore recursive application will make the
entire window $(m, M+24)$-guillotine.

We have to show that this procedure terminates. But this follows from
the fact that all cuts are weakly central: Each recursion step reduces
one coordinate of the window by at least 1 or $1/4$ of its
width/height. We only add new edges in the interior of a subwindow
when the $M+24$-region-span is non-empty, therefore, for small enough
windows, we do not add edges to the interior of their subwindows. At
that point, the minimum length, width and height of an edge inside of
the window remains fixed, so at some point, all edges that lie
completely in the interior of $W$ will be axis-parallel, and as soon
as one coordinate gets small enough, all of them are parallel. But
then we can cut between them, so each window only contains one such
edge. And lastly, we can cut that edge in half. (There is probably a
simpler argument.) 

There are separate arguments for the length charged to disks and
edges. For each cut $l$, we have added a length of $\sigma_m(l) +
\Sigma_M(l)$. We can charge it off as follows: There are $2m$ edge
segments making a segment of length $\delta_m(l)$ $m$-dark. Charge
each of them with $8 \cdot 1/2m$. (More precisely, these are not
necessarily the same edges or even connected, but there are edges of
total length at least $2m \delta_m(l)$ within the same window such
that their orthogonal projection onto $l$ intersects at most $m-1$
other edges in $E' \cup E^*$. We can charge all of those with $8 \cdot
1/2m$.)

Similarly, there are disk boundary segments of total length $M \cdot
\Delta_{M/2}(l)$ making parts of the cut $M/2$-region-dark. We can
charge each of them with $8/M$. 

In total, the charged length is $M \cdot \Delta_{M/2}(l) \cdot 8/M + 2m
\cdot \delta_m(l) \cdot 8 \cdot 1/2m = 8 (\Delta_{M/2}(l) + \delta_m(l))
\geq \sigma_m(l) + \Sigma_M(l)$, because the cut is $8$-favorable.  

We also know that $\Sigma_{M+24}(l) + 2 \leq \Sigma_{M}(l)$, because
the additional segments in $\Sigma_M(l)$ both visit 12
disks. Therefore, $\sigma_m(l) + \Sigma_M(l) \geq \sigma_m(l) +
\Sigma_{M + 24}(l) + 2$, which is an upper bound on the actual length
of what we insert -- $m$-span, $(M+24)$-region-span, and possibly a
connecting segment of length at most 2.

So the charge indeed will be an upper bound on the additional length
as described in the beginning of the proof.

It remains to show that each segment is charged directly only a
constant number of times, and that the total indirect charge is
sufficiently small. 

For edge segments, there are at most $m-1$ other segments between a
charged segment $e$ and the cut $l$. But the cut then becomes the
boundary of the next subwindow. Therefore, this edge will not make a
cut $l'$ between $e$ and $l$ $m$-dark, and not be charged again from
this direction. Since all cuts are axis-parallel, each edge is indeed
only charged at most 4 times.

For disks, the same argument works: Only the $M/2$ disks closest to a
cut (and making it $M/2$-dark or even $M$-dark) in a given direction
are charged, and each disk can only be among those and make the cut
$M/2$-dark once for each direction. 

Finally, we have to take care of indirect charge. This applies to both
disks and edges, because while our analysis shows that we can upper
bound the additional length for the connecting segments by the
$M$-region-span, the length of this span may be accounted for by
$m$-dark and not by $M/2$-region-dark segments. 

But for each edge segment, the direct charge is at most $16/m$. The
segment of that length might be charged with $16/m$ again, adding a
charge of $(16/m)^2$ to the original segment, yielding a geometric
series. Therefore, the total charge is at most $\frac{16}{m-16} \leq
32/m$ for $m \geq 32$.

For disks, the same analysis works: The direct charge is at most
$8/M$, and since $m \geq 32$, the indirect charge is bounded by $8/M
\cdot 32/m \leq 8/M$. 

Finally, we duplicate each new edge segment to make the resulting
graph Eulerian, increasing the additional length by a factor of
2. This concludes the proof.
\end{proof}

We did not use the fact that the regions are unit disks: It it
sufficient to assume they are disk-like and modify the constants
accordingly. 

\subsection{Grid}
By computing a $(1+\vare)$-approximation of a tour visiting the
centers of the disks, we obtain a TSPN tour which is at most an
additive $2k (1+ \vare)$ from the optimum (for $k$ large, this is a
constant-factor approximation algorithm and was analyzed in
\cite{dumi}). If the tour has length at least $\frac{2k
  (1+\vare)}{\vare}$, this is a sufficiently good solution; otherwise,
the disks are within a square of size $\lceil 3k/\vare \rceil \times
\lceil 3k/\vare \rceil$, if $\vare \leq 1/3$ and $k \geq 6$.

Such a square can be found (or shown that none exists) in polynomial
time. We then equip it with a regular rectilinear grid with edge
length $\delta := (2 \lceil k/\vare \rceil)^{-2}$. 

\begin{definition} \label{def:grid}
  An edge set is \textbf{grid-rounded} w. r. t. a grid $G$, if all
  edge endpoints are on the grid. A polygon is grid-rounded, if its
  boundary is grid-rounded. A set of regions is grid-rounded, if all
  regions are grid-rounded polygons. 

  A coordinate (or axis-parallel line segment) is said to be a
  \textbf{half-grid coordinate} of $G$, if it is on the grid or in
  the middle between two consecutive grid points.
\end{definition}

At a cost of a factor $(1+\vare)$, the instance can be grid-rounded,
such that every disk center is on a grid point. The same can be done
for an optimum solution, i. e. every edge should begin and end in a
grid point. Such a solution is still feasible, if we replace each disk
$D_i \in \mathcal{D}$ by the convex hull of the set of grid points
$\Gamma_i$ it contains. As this convex hull contains at least the
diamond inscribed in the disk (a square of area 2), and it has
rotational symmetry, it will still be $\alpha$-fat$_E$, for slightly
smaller $\alpha$. The previous theorem still holds for these regions
(with modified constants). 

The regions $\conv(\Gamma_i)$ are polygonal. Therefore, in
Lemma~\ref{lem:wfav-cut}, the functions $\delta_m, \Delta_M, \sigma_m$
and $\Sigma_M$ are piecewise linear, with discontinuities at grid
points. This implies:

\begin{lemma}   
  \label{lem:grid-cut} Given grid-rounded disk-like regions and a
  window $W$ with half-grid coordinates, there is a
  perfect cut with half-grid coordinates.
\end{lemma} 

\begin{proof}
  This follows from Lemma~\ref{lem:wfav-cut}, because the integrals
  involved can be replaced by ($\delta$ times) the sums of the
  function values at all half-grid points (that are not grid
  points). For piecewise linear functions, this sum is equal to the
  integral.
\end{proof}

If the $m$-span and $M$-region-span are inserted at half-grid cuts
(with their endpoints not even at half-grid points), the solution does
not remain on the grid. In particular, the connecting segment for the
$M$-region-span could lead to discontinuities in $\sigma_m(l)$ and
$\delta_m(l)$ at non-grid points, thus preventing the recursive
application of this lemma. Therefore, it will be moved to the grid in
the following. 

We can assume that every disk is visited by the endpoint of an edge
that lies on the grid. This costs a factor of $(1+\vare)$ and means
that when restructuring the edge set, as long as the resulting graph
remains Eulerian, connected and has the same (or more) edge endpoints,
we can still extract at TSPN tour from it. Therefore, the following
lemma can be applied without taking the regions into account:

\begin{lemma} \label{lem:span} Consider a regular rectilinear grid
  with edge length $\delta$ and a grid-rounded set $E$ of edges that
  is an optimum tour of the set of its edge endpoints, and a window
  $W$.

  Let $l$ be a cut in $W$, such that its $m$-span is non-empty and $l$
  has distance at least $\delta$ from the boundary edges of the window
  that are parallel to $l$.

  If the $x$- or $y$-coordinate of $l$ (depending on its
  orientation) is a grid coordinate, and $m$ intersects at least 15
  different edges in their interior (or has 16 intersection points
  with edges), then $\sigma_m(l) \geq \delta$. 

  Furthermore, if $l$ is in the center between two consecutive grid
  coordinates, then $\sigma_m(l) \geq \delta$, if its $m$-span
  intersects at least 19 different edges (necessarily in their
  interior, because their endpoints can only be on the grid).
\end{lemma}

This is a grid-rounded version of Arora's patching lemma~\cite[Lemma
3]{arora}. The proof uses similar ideas and exploits the grid structure to
show that in some cases, the patching construction actually decreases
the tour length. 

\begin{proof}
  First, let $l$ be on the grid and without loss of generality
  vertical. If $l$ has 16 intersection points with edges, then either
  two of them are (different) grid points, and $\sigma_m(l) \geq
  \delta$, or at $l$ intersects at least 15 different edges in their
  interior, therefore it is sufficient to consider this case. 

  If $\sigma_m(l) < \delta$, this configuration can never be
  optimal. This follows form the construction in
  Figure~\ref{pic:patching}: Expand the $m$-span to the nearest grid
  coordinates above and below, and then consider the box that has as
  left and right edge this expanded $m$-span translated by $\pm
  \delta$. This box will have height $\delta$ or $2\delta$ and width
  $2\delta$. If it has height $2\delta$, an additional length of $2
  \delta$ is used to connect to the grid point in the center of the
  box.

  \begin{figure}[hbt] \center
    \includegraphics[width=4cm]{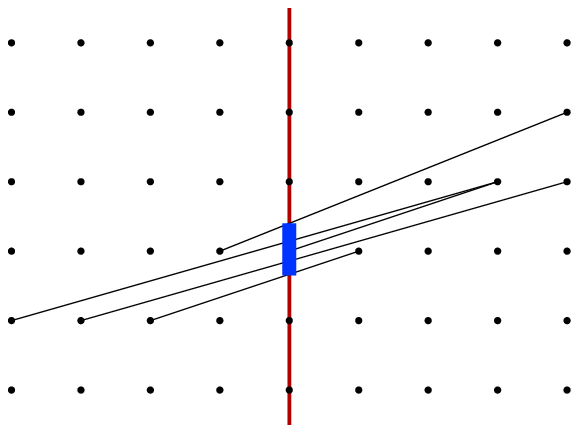} \hspace*{0.3cm}
    \includegraphics[width=4cm]{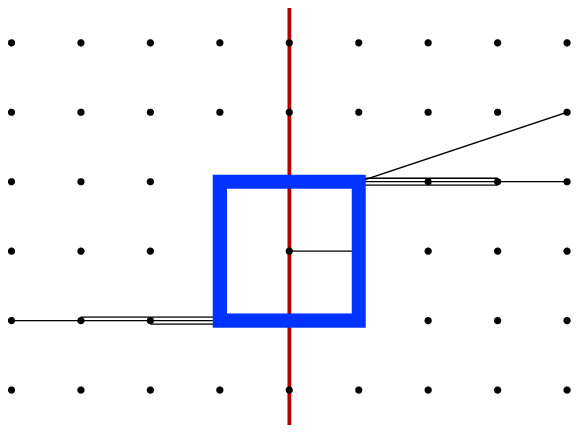}
    \caption{Before and after applying the grid patching lemma to grid
      cuts}
    \label{pic:patching}
  \end{figure}

  For every edge intersecting the box, split it into different parts
  at the intersection points. The inner part (inside of the box) has
  length at least $\delta$, but it does not visit any new endpoints,
  and can therefore be removed. For the other parts, if we consider
  their second endpoint (not on the boundary of the box) fixed and
  choose their first endpoint among the points of the boundary of the
  box, they become shortest possible when connected in such a way,
  that the endpoint is a vertex of the box or the segment is
  orthogonal to the boundary edge of the box. In both cases, the first
  endpoint is a grid point. Therefore, for these parts there is a grid
  point on the boundary of the box, such that the intersection point
  can be moved there without increasing the length of the edge set. 

  This preserves connectivity and parity except possibly on the
  boundary of the box. Since it has perimeter $\leq 8 \delta$, edges
  of length $\leq 4 \delta$ can be used to correct parity. 

  Overall, we have added edges of total length $8 \delta
  + 2 \delta + 4 \delta = 14 \delta$ and removed edges of length $\geq
  \delta$ for every edge intersecting $\sigma_m(l)$ in its
  interior. Hence, for an optimal tour, there can be at most 14 such
  edges.

  If $l$ is not on the grid, but at a half-grid coordinate, we can
  apply a similar argument and construction, see
  Figure~\ref{pic:patching2}.

  \begin{figure}[hbt] \center
    \includegraphics[width=4cm]{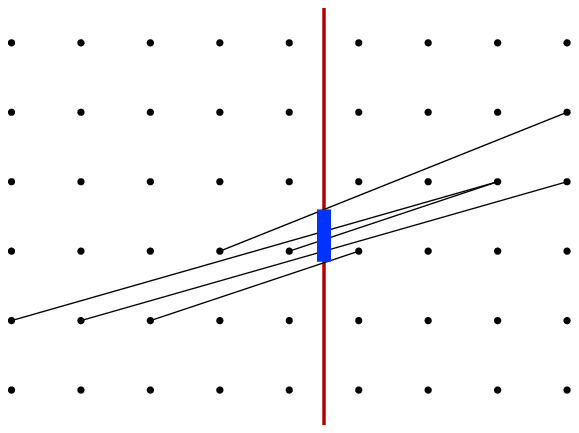} \hspace*{0.3cm}
    \includegraphics[width=4cm]{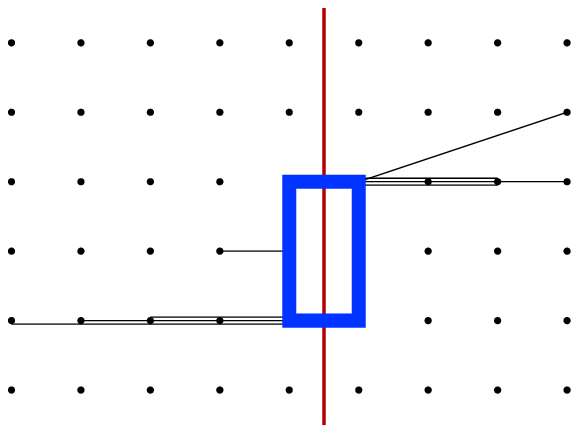}
    \caption{Before and after applying the grid patching lemma to
      half-grid cuts}
    \label{pic:patching2}
  \end{figure}

  The box has perimeter $\leq 6 \delta$ and for every edge
  intersecting it, the inner segment of length at least $\delta/2$ can
  be removed. There are no interior points to be visited, and
  correcting parity costs at most $3 \delta$. Therefore, there can be
  at most $\frac{3 \delta + 6 \delta}{\delta/2} = 18$ edges
  intersecting the $m$-span, if its length is less than $\delta$. 

\end{proof}

For the $M$-region-span, no such lemma is needed, because if we insert
it, we can also afford a connecting segment of length $2$. Therefore,
it can be extended to the grid without increasing the length we used
for the entire construction by more than a factor of $2$ (and
actually, $1 + 2/\delta$) -- provided the cut is on the grid.

Choosing only cuts on the grid is not sufficient, as the following
example shows: Even without regions, there is no $1$-good cut with
grid coordinates. 

\begin{figure}[hbt] \center
  \includegraphics[width=2cm]{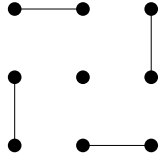}
  \caption{No 1-good cut with grid coordinates}
  \label{pic:grid-cut}
\end{figure}

The recursive construction in Theorem~\ref{thm:guillotines} makes cuts
$m$-good and $M$-region-good by inserting edges on the cut, which is
not possible for a cut that is not at a grid coordinate. Since we
cannot change the position of the regions, no modification of the edge
set (preserving grid-roundedness) can make the cut
$M$-region-good. Moving the cut to the grid is also not an option,
since Figure~\ref{pic:grid-cut} shows that this is not always possible.

Therefore, we cannot hope to find an $(m,M)$-guillotine subdivision
with this construction. However, for algorithmic purposes, the main
aim of the guillotine property was avoiding the enumeration of which
subproblem is responsible for visiting which regions. What happens, if
we make them \emph{both} responsible for all regions in the
$M$-region-span?

First, note, that ``smaller'' subproblems never rely on the fact that
their containing windows are guillotine, because they do not yet know
the respective cuts -- with the exception of the four cuts defining
their boundaries. For these cuts, we can easily enumerate the
possibilities ``visit all regions in the $M$-region-span'' and ``visit
none of them''. 

Intuitively, it seems that this construction would significantly
increase the length of the subdivision. But we know that those regions
can be visited by a segment with at most the length of the
$M$-region-span. More importantly, we know that they can be visited on
each side of the cut by a segment such that their combined length is
at most $2 \Delta_M(l)$, which we can afford by the charging
scheme. Not both of these can necessarily be connected to $E$ within
their containing window (but at least one), therefore we should add
(or at least ``reserve'') an edge across the cut, i. e. only obtain an
$(m+1, M+24)$-guillotine subdivision.

To accommodate these changes, we redefine $M$-region-good and thus
obtain a new $(m,M)$-guillotine property: 
\begin{definition}
  A cut $l$ is \textbf{$\boldsymbol{M}$-region-good} with respect to
  $W$, $\mathcal{R}_{W_0}$ and $E$, if there are no two regions $R, R'
  \in \mathcal{R}$ that intersect the $M$-region-span of $l$, but $E$
  visits $R$ and $R'$ only on different sides of the cut. 
\end{definition}

In other words, if $l$ does not visit $R$ on one side of the cut, it
must visit all other regions that intersect the $M$-region-span on the
other side of the cut. Here, ``side of the cut'' denotes a closed
half-space; in particular, a cut that was $M$-region-good w. r. t. the
previous definition remains $M$-region-good, since the $M$-region-span
is in $E$ and thus $E$ visits all regions in the $M$-region-span on
both sides of the cut. 

The following two lemmas show that this construction works -- for both
edges and regions, we can replace the operations ``insert the span''
by one of the following in the charging scheme, and get the statement
of the charging scheme (with modified constants) for grid-rounded
subdivisions. 

\begin{lemma} \label{lem:m-good} Given a perfect
  half-grid cut $l$ in a window $W$ of width $\geq \delta$, and a
  connected Eulerian grid-rounded edge set $E$, then there is a
  grid-rounded edge set $E'$ that differs from $E$ only at edges
  intersecting $W$, has the $(m', M')$-guillotine property outside of
  $W$ if $E$ does, is by at most an additive
  $\mathcal{O}(\sigma_m(l))$ longer than $E$, visits the same (or a
  superset) of the grid points $E'$ visits, and is connected and
  Eulerian, such that $l$ is $(m+9)$-good w. r. t. $E'$ and $W$.
\end{lemma}

If the window has width $< \delta$, then the constructions in the
proof might not be inside of the window. On the other hand, such a
window is $(m,M)$-guillotine by definition, since it cannot contain
(the entire interior of) grid edges in its interior.

\begin{proof}
  Without loss of generality, let $l$ be a vertical cut. If $l$ is
  $(m+9)$-good, there is nothing to show. Otherwise, there are at
  least $19$ edges on the $m$-span, so it has length at least $\delta$
  by Lemma~\ref{lem:span}, or $E$ can be made shorter by applying the
  construction there, thereby making the cut $(m+1)$-good. 

  If $\sigma_m(l) \geq \delta$, there are two cases: If $l$ is not on
  the grid, as in Figure~\ref{pic:half-cut}, we insert an ``H'' shape,
  which has length $\leq 5 \delta + 2 \sigma_m(l) \in
  \mathcal{O}(\sigma_m(l))$. For all edges intersecting this H, the
  intersection point should be moved to a grid point without
  increasing the length or violating guillotine property.

  \begin{figure}[hbt] \center
    \includegraphics[width=3.5cm]{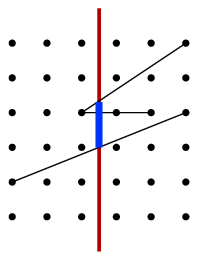}
    \hspace*{1cm}
    \includegraphics[width=3.5cm]{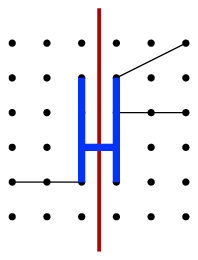}
    \caption{Construction for half-grid cuts}
    \label{pic:half-cut}
  \end{figure}

  To see that the length does not increase, let $p,q$ be the endpoints
  of an edge intersecting the H, then replacing the edge by segments
  from $p$ to the first intersection point $p_H$ and from the last
  intersection point $q_H$ to $q$ preserves connectivity (and parity
  can be correcting using edges of the H). Let $p$ be to the left of
  $l$, then $p_H$ is either the left vertical edge or the bar. In the
  latter case, moving the intersection point to the left endpoint of
  the bar only decreases the length of the segment. In the former
  case, $p_H$ is either a grid point or can be moved
  vertically on the H. In that case, either the edge from $p$ to $p_H$
  is horizontal (so $p_H$ is a grid point because $p$ is), or there is
  a direction such that the angle at $p_H$ gets less acute when moving
  $p_H$, thereby making $(p, p_H)$ shorter. Whenever $(p, p_H)$
  intersects a grid point, we subdivide it and continue with the
  segment containing $p_H$, thus ensuring that the edge set remains
  planar an $(m', M')$-guillotine, if $E$ is. 

  The resulting graph is grid-rounded and $l$ is $m$-good, since the
  $m$-span only contains one point, and this point is part of the edge
  set. It might not be Eulerian, but the only points whose parity
  might have changed are on the H, hence duplicating some of its edges
  will make the edge set Eulerian again. 
  
  \begin{figure}[hbt] \center
    \includegraphics[width=3.5cm]{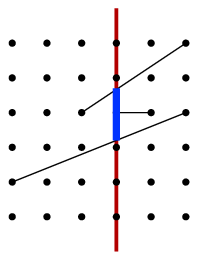}
    \hspace*{1cm}
    \includegraphics[width=3.5cm]{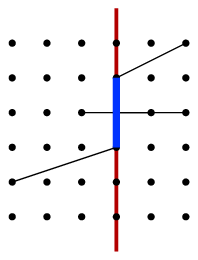}
    \caption{Construction for grid cuts}
    \label{pic:span-again}
  \end{figure}

  If the cut is at a grid coordinate, we increase the length of the
  $m$-span by at most $2 \delta$ as shown in
  Figure~\ref{pic:span-again} (so that it becomes grid-rounded) and
  then proceed analogously to the first case for all edges
  intersecting it.

\end{proof}

\begin{lemma} \label{lem:region-good} Given a perfect
  half-grid $(m+9)$-good cut $l$ in a window $W$, and a connected
  Eulerian grid-rounded edge set $E$ and grid-rounded disk-like
  regions $\mathcal{R}$, then there is a grid-rounded edge set $E'$
  that differs from $E$ only in edges intersecting of $W$, is $(m',
  M')$-guillotine outside $W$ if $E$ is, is by at most an additive
  $\mathcal{O}(\Sigma_M(l))$ longer than $E$, visits the same (or a
  superset) of the grid points $E'$ visits, and is connected and
  Eulerian, such that $l$ is $(m+10)$-good and $(M+C)$-good
  w. r. t. $E'$, $\mathcal{R}$ and $W$. The constant $C$ here depends
  on the constants for the disk-like regions
  (Definition~\ref{def:grid}) and can be chosen as 24 for unit disks.
\end{lemma}

\begin{proof}
  Without loss of generality, let $l$ be vertical. If $l$ is at a grid
  coordinate, extend $\Sigma_{M + C}(l)$ to the grid and move all
  intersection points with edges to the grid as in previous
  lemmas. Correct parity on the extended $(M+C)$-region-span.

  Otherwise, since the regions are polygonal, the set of regions in
  the $(M+C)$-region-span can be visited by a vertical segment of at
  most the same length either at the grid coordinate directly to the
  left or to the right of $l$. Insert this segment and possibly a
  connecting segment to the edge set (which might cross the cut)
  inside of $W$. For all edges intersecting the new vertical segment,
  proceed as before. For the connecting segment, note that it is not
  necessarily rectilinear. Therefore, if it intersects an edge, we can
  subdivide this edge and move the intersection point to a grid
  point. This costs at most $3 \delta$ , because there are 3 incident
  edges, and connects to the edge set -- hence it is sufficient to do
  this at most once. Since $\Sigma_M(l) \geq 2$ by choice of $C$, this
  is not too expensive, and can be done inside $W$. Again, correcting
  parity is only necessary on new segments. 
\end{proof}

Applying these lemmas to the charging scheme yields that there is an
$(m,M)$-guillotine grid-rounded subdivision that approximates a tour
well, more precisely:

\begin{theorem}
  \label{thm:grid}
  Let $\vare > 0$. For every set $\mathcal{D}$ of $k \geq 8$ disjoint
  unit disks within a square of size $\lceil 3k/\vare\rceil \times
  \lceil 3k/\vare \rceil$, let $m = \max\set{\lceil \frac{1}{\vare}
    \rceil, 8}$ and $M = \max \set{\lceil \frac{1}{\vare} \log_2
    (\frac{k}{\vare}) \rceil, 32}$.
  
  Then, there is an edge set $E$ with the following properties: 
  \begin{enumerate}
  \item It satisfies the (new) $(m+9,M+24)$-guillotine property.
  \item The endpoints of every edge are on a regular rectilinear grid
    with edge length $\delta = (2\lceil k/\vare \rceil)^{-2}$.
  \item It visits at least one point from each of $\Gamma_1, \dots,
    \Gamma_k$, the grid points of the slightly perturbed, polygonal
    approximations of the disks.
  \item It is Eulerian and connected. 
  \item The total length of all its segments is
    $(1+\mathcal{O}(\varepsilon))L^*$, where $L^*$ is the length of an
    optimum tour visiting $\mathcal{D}$ (or $\Gamma_1, \dots,
    \Gamma_k$, as both lengths only differ by a factor of at most $1 +
    \vare$).
  \end{enumerate}
\end{theorem}

This theorem implies an approximation ratio of
$(1+\mathcal{O}(\vare))$ for Mitchell's dynamic programming algorithm
for the TSPN with disjoint unit disks, if the grid and $m$ and $M$ are
chosen as above, and for these parameters, such a subdivision can be
found by Mitchell's algorithm (together with the refinements to
preserve grid-roundedness) in polynomial time
(Theorem~\ref{thm:main}). 

\begin{proof}[Theorem~\ref{thm:grid}]
  Using the two previous lemmas, one can construct $E'$ as follows:
  Starting with an optimal tour on the grid, recursively find a
  perfect half-grid cut, insert edges so that it becomes $(m+9)$-good
  and $(M+24)$-region-good, and continue with the new subwindows. The
  edge set remains a tour, and becomes $(m+9, M+24)$-guillotine. The
  increase in length can be bounded using the same charging scheme as
  in Theorem~\ref{thm:guillotines}.
\end{proof}

\section{Conclusion}
The guillotine subdivision method of Mitchell \cite{mitchell-tsp,
  mitchell-ptas} can be used to derive a PTAS for the TSP with unit
disk neighborhoods. All arguments carry over to disk-like regions, for
which Mitchell's framework can be used to derive a PTAS as well. This
includes geographic clustering as the special case when the regions
are $\alpha$-fat (they could be $\alpha$-fat\E instead).

However, the approach of Bodlaender et al. \cite{grigoriev} based on
curved dissection in Arora's PTAS for TSP \cite{arora} achieves a
faster theoretical running time for disjoint connected
regions with geographic clustering. Their algorithm, like Arora's, can
be generalized to more than two dimensions. 

For $\alpha$-fat\E regions, the best known results are the constant-factor
approximation algorithm of \cite{elb} and the QPTAS of
\cite{elb-qptas}. 

For $\alpha$-fat regions in Mitchell's sense, the existence of a PTAS
remains open. The problem with external regions of
Section~\ref{sec:ext} can be avoided by using $\alpha$-fat\E or convex
regions instead, the charging scheme (Section~\ref{sec:charge}) can be
fixed by bounding the ratio of perimeter and diameter, the grid can be
handled as in the unit disk case, and localization does not require
any additional assumptions. However, even for those stronger
conditions on the regions, it is unclear how to handle connectivity
(Section~\ref{sec:cnn}) for neighborhoods of varying size. The length
of the connecting segment can be bounded for $\alpha$-fat\E regions as
in the unit disk case, but it might still destroy the guillotine
property of other windows. To our knowledge, no PTAS for any form of
the TSP with neighborhoods of varying size exists.

Mitchell's constant factor approximation algorithm for disjoint
connected regions \cite{mitchell-cfa} relies on the PTAS for
$\alpha$-fat regions, but only applies it to disjoint balls, which are
$\alpha$-fat$_E$. Therefore, a constant factor approximation algorithm
by Elbassioni et al. \cite{elb} can be used instead, so that the
overall algorithm in \cite{mitchell-cfa} still works and yields a
constant factor approximation for the TSP with general disjoint
connected, and in particular $\alpha$-fat, regions.

\end{document}